\documentclass[a4paper,final]{article}
\usepackage{graphicx,amssymb,amsmath}
\pdfoutput=1

\usepackage{fullpage}

\newcommand{\CH}{\operatorname{CH}}
\DeclareMathOperator{\rcr}{\overline{cr}}

\usepackage{xcolor}

\newcommand{\LF}{\left\lfloor}
\newcommand{\RF}{\right\rfloor}
\newcommand{\LC}{\left\lceil}
\newcommand{\RC}{\right\rceil}

\usepackage{amsthm}      % for proofs (and theorems)
\newtheorem{theorem}{Theorem}
\newtheorem{lemma}[theorem]{Lemma}
\newtheorem{cor}[theorem]{Corollary}
\newtheorem{obs}[theorem]{Observation}
\newtheorem{prop}[theorem]{Proposition}
\newtheorem{conj}[theorem]{Conjecture}
\title{On $k$-Gons and $k$-Holes in Point Sets\thanks{%
Research of Oswin Aichholzer and \mbox{Birgit} Vogtenhuber supported by the ESF EUROCORES programme EuroGIGA -- CRP `ComPoSe', Austrian Science Fund (FWF): I648-N18.
Research of Ruy Fabila-Monroy partially supported by CONACyT (Mexico), grant 153984.
Research of Thomas Hackl supported by the Austrian Science Fund (FWF): P23629-N18 `Combinatorial Problems on Geometric Graphs'.
Research of Clemens Huemer partially supported by projects MTM2012-30951 and Gen. Cat. DGR 2009SGR1040. 
Research of Marco Antonio Heredia, Hern{\'a}n Gonz{\'a}lez-Aguilar, and Jorge Urrutia partially supported by CONACyT (Mexico), grant CB-2007/80268.
Work by Pavel Valtr supported by project 1M0545 of the Ministry of Education of the Czech Republic. 
}}
\author{Oswin Aichholzer\thanks{Institute for Software Technology,
        University of Technology, Graz, Austria,
        {\tt [oaich|thackl|bvogt]@ist.tugraz.at}}
		\and
		Ruy Fabila-Monroy\thanks{Departamento de Matem{\'a}ticas,
        Cinvestav, D.F. M\'exico, M{\'e}xico,
        {\tt ruyfabila@math.cinvestav.edu.mx}}
        \and 
        Hern{\'a}n Gonz{\'a}lez-Aguilar \thanks{Facultad de Ciencias, 
		Universidad Aut\'onoma de San Luis Potos\'i, San Luis Potos\'i, M\'exico,
		{\tt hernan@fc.uaslp.mx}}
        \and
        Thomas Hackl$^\dag$%\thanks{Institute for Software Technology,
		\and  
        Marco A. Heredia %Marco Antonio Heredia Velasco 
		\thanks{ % Mexico, {\tt marco.heredia@gmail.com}}
		Departamento de Sistemas, 
		Universidad Aut\'onoma Metropolitana - Azcapotzalco, 
		D.F. M\'exico, M\'exico, 
		{\tt  marco@ciencias.unam.mx}} 
		\and  
        Clemens Huemer\thanks{Departament de Matem{\`a}tica Aplicada IV, 
          Universitat Polit{\`e}cnica de Catalunya, Barcelona, Spain, 
          {\tt clemens.huemer@upc.edu}}
		\and 
        Jorge Urrutia\thanks{Instituto de Matem{\'a}ticas,
        Universidad Nacional Aut{\'o}noma de M{\'e}xico, D.F. M\'exico, M\'exico,
		{\tt urrutia@matem.unam.mx}}
		\and
		Pavel Valtr\thanks{Department of Applied Mathematics and Institute for Computer Science (ITI), 
		Charles University, Prague, Czech Republic} %% no email wanted !! {\tt valtr@kam.mff.cuni.cz}}
		\and
        Birgit Vogtenhuber$^\dag$%\thanks{Institute for Software Technology,
}

\index{Vogtenhuber, Birgit}
\index{Huemer, Clemens}
\index{Fabila-Monroy, Ruy}
\index{Aichholzer, Oswin}
\index{Heredia, Marco A.}
\index{Urrutia, Jorge}
\index{Hackl, Thomas}
\index{Gonz{\'a}lez-Aguilar, Hern{\'a}n}
\index{Valtr, Pavel}

\begin{document}
%------------------------------ Make Title -------------------------------
\thispagestyle{empty}
\maketitle 
%------------------------------ Text -------------------------------------
\begin{abstract} 
  We consider a variation of 
  the classical Erd\H{o}s-Szekeres problems %%
  on the \mbox{existence} and number of convex \mbox{$k$-gons} and \mbox{$k$-holes}
  (empty \mbox{$k$-gons}) in a set of $n$ points in the plane.
  Allowing the \mbox{$k$-gons} to be non-convex,
  we show bounds and structural results on 
  maximizing and minimizing their numbers.
  Most noteworthy, for any $k$ and sufficiently large $n$, we give a 
  quadratic lower bound for the number of 
  \mbox{$k$-holes}, and show that this number is maximized by sets in convex 
  position. 
\end{abstract}

%------------------------------ Introduction -----------------------------
\section{Introduction}
\label{sec:intro}

Let $S$ be a set of $n$ points in general position in the plane 
(i.e, no three points of $S$ are collinear). 
A \mbox{$k$-gon} is a simple polygon spanned by $k$~points of $S$.  
A \mbox{$k$-hole} is an empty \mbox{$k$-gon}; that is, a 
\mbox{$k$-gon} that contains no points of $S$ in its interior. 

Around 1933 Esther Klein raised the following question, which was
(partially) answered in the classical paper by Erd\H{o}s and
Szekeres~\cite{ES} in 1935: ``Is it true that for any $k$ there is a
smallest integer $g(k)$ such that any set of $g(k)$ points contains at
least one convex \mbox{$k$-gon}?'' 
As observed by Klein, $g(4)=5$, and
Kalbfleisch et al.\!~\cite{KKS}  proved that %solved the more involved case of
$g(5)=9$. The case \mbox{$k=6$}
was only solved as recently as 
2006 by Szekeres and Peters~\cite{SP}.
They showed that $g(6)=17$ by an exhaustive computer search.
The well known Erd\H{o}s--Szekeres Theorem~\cite{ES} 
states that $g(k)$ is finite for any $k$. 
The current best bounds are $2^{k-2}+1 \leq g(k) \leq \binom{2k-5}{k-2} +1$ for $k\geq 5$,
where the lower bound
goes back to Erd\H{o}s and Szekeres~\cite{ES2} and is conjectured to
be tight.
There have been many improvements on the upper bound, where
the currently best bound has been obtained in 2005 by T\'oth and 
Valtr~\cite{TV2}; see e.g.~\cite{a-ecgrr-09} for more details.

Erd\H{o}s and Guy~\cite{EG} posed the following generalization: 
``What is the least number of convex $k$-gons determined by any set $S$ of
$n$ points in the plane?'' The trivial solution for the case $k=3$ is
$\binom{n}{3}$.  
For convex \mbox{4-gons} this question is highly non-trivial, as it is
related to the search for the rectilinear crossing number
$\rcr(n)$, the minimum number of crossings in a straight-line drawing of the 
complete graph with $n$ vertices; see the next section for details.

In 1978 Erd\H{o}s~\cite{Er78} raised the following question for convex
\mbox{$k$-holes}: ``What is the smallest integer $h(k)$ such that any set of
$h(k)$ points in the plane contains at least one convex \mbox{$k$-hole}?''
As observed by Esther Klein, every set of 5 points determines
a convex \mbox{4-hole},
and Harborth~\cite{Ha78} showed that 10 points always contain a convex \mbox{5-hole}.
Surprisingly, in 1983 Horton showed
that there exist arbitrarily large sets of points containing no
convex \mbox{7-hole}~\cite{Ho83}. 
It took almost a quarter of a century after Horton's construction to answer the existence 
question for $k=6$.
In 2007/08 Nicol\'{a}s~\cite{nic} and independently
Gerken~\cite{gerk} proved that every sufficiently large point set
contains a convex \mbox{6-hole}; see also~\cite{V08}. 

A natural generalization of the existence question for $k$-holes is 
this:
``What is the least number $h_k(n)$ of
\mbox{convex} \mbox{$k$-holes} determined by any set of $n$ points in
the plane?''  
Horton's construction implies 
$h_k(n)=0$ for $k \geq 7$. Table~\ref{tab:holes2} shows the
current best lower and upper bounds for $k=3, \ldots, 6$.

\begin{table}[htb]
\begin{center}
{%\small
\setlength{\tabcolsep}{3pt}
\renewcommand{\arraystretch}{1.40}
\begin{tabular}{rcl}
\hline
\vspace{-10pt}\\
%% TODO NEWEST RESULTS!!
$n^2 - \frac{32}{7}n + \frac{22}{7}$  	& $\leq h_3(n) \leq$ & $1.6196n^2+o(n^2)$ \\
$\frac{n^2}{2} - \frac{9}{4}n - o(n)$   & $\leq h_4(n) \leq$ & $1.9397n^2+o(n^2)$ \\
$\frac{3n}{4} - o(n)$  	                & $\leq h_5(n) \leq$ & $1.0207n^2+o(n^2)$ \\
$\frac{n}{229}- 4$					 	& $\leq h_6(n) \leq$ & $0.2006n^2+o(n^2)$ \\
\vspace{-10pt}\\
\hline
\end{tabular}\vspace{-2ex}
}
\end{center}
\caption{Bounds on the numbers $h_k(n)$ of convex \mbox{$k$-holes}~\cite{afhhpv-lbnsc-14,BV2004,v-ephpp-12}.}
\label{tab:holes2}
\end{table}

All upper bounds in the table are due to B{\'a}r{\'a}ny and Valtr~\cite{BV2004}. 
They are obtained by improving constructions that had been developed by Dumitrescu~\cite{Du} and previously improved by Valtr~\cite{Va95}.

Concerning the lower bounds for $k \leq 5$, Dehnhardt~\cite{De87} showed in his PhD thesis that for $n\geq 13$,
$h_3(n) \geq n^2 - 5n + 10$,
$h_4(n) \geq \binom{n-3}{2} + 6$, and 
$h_5(n) \geq 3\LF\frac{n}{12}\RF$. 
As this PhD thesis was published in German and is not easy to access, later on 
several weaker bounds have been published. 
Only very recently these results have been subsequently 
improved~\cite{g-nnet-11, g-nnet-12, 5HOLES_EGC, ahv-5g5h-12, v-ephpp-12, afhhpv-lbnsc-14}, 
where the currently best bounds can be found in~\cite{afhhpv-lbnsc-14}, using a remarkable result from~\cite{g-nnet-12}.
A result of independent interest is by Pinchasi~et~al.\!~\cite{PiRaSh06}, 
who showed $h_4(n) \geq h_3(n) - \frac{n^2}{2} - O(n)$ and $h_5(n) \geq h_3(n) - n^2 - O(n)$. 
By this, 
improving the $n^2$-factor in the lower bound of $h_3(n)$ 
implies better lower bounds also for $h_4(n)$ and $h_5(n)$.
Concerning lower bounds on the number of 6-holes, a proof of $h_6(n) \geq  \lfloor \frac{n-1}{858} \rfloor-2$ 
is contained in the proceedings version~\cite{KHOLES_CCCG} of the paper at hand.
This proof is based on $h_6(1717)\geq 1$ by Gerken~\cite{gerk}.
Valtr~\cite{v-ephpp-12} presented an improved bound of $h_6(n) \geq  \frac{n}{229}-4$, 
combining a different proof technique with Koshelev's result of $h_6(463)\geq 1$~\cite{ko-espeh-2009}.
As combining this result by Koshelev with the proof of~\cite{KHOLES_CCCG} only gives 
$h_6(n) \geq  \frac{n}{231}-O(1)$, the according proof is omitted here. %in this version.
%
%%%%%%%%%%%%%%%%%%%%%%%%%%%%%%%%%%%%%%%%%%%%%%%%%%%%%%%%%%%%%%%%%%%

In this paper we generalize the above questions on the numbers of \mbox{$k$-gons} and \mbox{$k$-holes} 
by allowing the gons/holes  to be non-convex. 
Thus, whenever we refer to a (general) \mbox{$k$-gon} or \mbox{$k$-hole}, 
unless it is specifically stated to be convex or non-convex, it could be either. 
Similar results for \mbox{4-holes} and \mbox{5-holes} can be found in~\cite{4HOLES_CGTA} 
and~\cite{ahv-5g5h-12}, respectively.  
The PhD thesis~\cite{mythesis} summarizes most results obtained for $k \geq 4$.
See also~\cite{a-ecgrr-09} for a survey on the history of questions and 
results about \mbox{$k$-gons} and \mbox{$k$-holes}.
We remark that in some related literature, \mbox{$k$-holes} are assumed to be convex.

A set of $k$ points in convex position spans precisely one
convex \mbox{$k$-hole}. In contrast, a point set might admit exponentially
many different polygonizations (spanning cycles)~\cite{GNT00}. 
Thus, the number of $k$-gons and \mbox{$k$-holes} 
can be larger than $\binom{n}{k}$, which makes the questions considered in this paper
more challenging (and interesting) than they might appear at first glance.

Tables~\ref{tab:overview_gons} and~\ref{tab:overview_holes} 
summarize the best current bounds
on the numbers of \mbox{$k$-gons} and \mbox{$k$-holes}, including the
results of this paper. 
The entries in the tables 
list lower and upper bounds, 
also in explicit form if available,
thus indicating for which values there are still gaps to close.
Among other results, we generalize properties concerning \mbox{4-holes}~\cite{4HOLES_CGTA} 
and \mbox{5-holes}~\cite{ahv-5g5h-12} to %% constant 
$k \geq 6$. 
In Section~\ref{sec:gons} we give asymptotic bounds on the number of 
non-convex and general \mbox{$k$-gons}.
In Section~\ref{sec:max_gen_holes} we consider (general) \mbox{$k$-holes}. We show that 
for sufficiently small $k$ with respect to $n$ their number is maximized by sets in convex position, 
which is not the case for large $k$.
Section~\ref{sec:ub_nc_holes} provides a tight bound for the maximum 
number of non-convex \mbox{$k$-holes}, 
and Section~\ref{sec:min_gen_holes} contains bounds for the minimum number of
general \mbox{$k$-holes}. We conclude with open problems in 
Section~\ref{sec:conclusion}.

\begin{table*}[htb]
\begin{center}
{%\small
\setlength{\tabcolsep}{4pt}
\setlength{\arraycolsep}{1pt}
\renewcommand{\arraystretch}{1.40}
\begin{tabular}{l||l|l|l|l}
      	& \multicolumn{1}{c|}{convex} & \multicolumn{1}{c|}{non-convex} & \multicolumn{2}{c}{\hspace{0.0cm}general} \\
		& \multicolumn{1}{c|}{min} & \multicolumn{1}{c|}{max} & \multicolumn{1}{c|}{min} & \multicolumn{1}{c}{max} \\ \hline\hline
%%%
$k\!=\!4$ 	& $\begin{array}{ll}  %convex min 4-gons
			\rcr(n) \\
			\Theta (n^4)
		   \end{array}$ 
	  	& $\begin{array}{l}  %non-convex max 4-gons
			3{n \choose 4}\!-\! 3\rcr(n) \\ 
			\Theta (n^4)~\cite{ahv-5g5h-12}
		   \end{array}$  
		& $\begin{array}{l}  %general min 4-gons
			{n \choose 4}  \\
			\Theta (n^4)~\cite{ahv-5g5h-12}
		   \end{array}$ 
		& $\begin{array}{l}  %general max 4-gons
			3{n \choose 4}\!-\! 2\rcr(n) \\
			\Theta (n^4)~\cite{ahv-5g5h-12}
		   \end{array}$								\\\hline
%%%%%%
$k\!=\!5$ 	& $\begin{array}{l}   %convex min 5-gons
			 \Theta (n^5)~\cite{BMP05} 
		   \end{array}$ 
	  	& $\begin{array}{l}   %non-convex max 5-gons
			 10{n \choose 5}\! -\! 2(n\!-\!4)\rcr(n) \\ 
			 \Theta (n^5)~\cite{ahv-5g5h-12}
		   \end{array}$  
		& $\begin{array}{l}   %general min 5-gons
			 {n \choose 5}  \\
			 \Theta (n^5)~\cite{ahv-5g5h-12}
		   \end{array}$ 
		& $\begin{array}{l}   %general max 5-gons
			 \Theta (n^5)~\mbox{[Sec.~\ref{sec:gons}]}
		   \end{array}$								\\\hline
%%%%%%
$k\!\geq\!6$ & $\begin{array}{l} %convex min k-gons
			 \Theta (n^k)~\cite{BMP05} 
		   \end{array}$ 
	  	& $\begin{array}{l}  %non-convex max k-gons
%			\red{\geq} \mbox{\red{\bf ???}}  \\
			 \Theta (n^k)~\mbox{[Sec.~\ref{sec:gons}]}
		   \end{array}$  
		& $\begin{array}{l}  %general min k-gons
			 {n \choose k}  \\
			 \Theta (n^k)~\mbox{[Sec.~\ref{sec:gons}]}
		   \end{array}$ 
		& $\begin{array}{l}  %general max k-gons
%			\red{\geq} \mbox{\red{\bf ???}}  \\
			 \Theta (n^k)~\mbox{[Sec.~\ref{sec:gons}]}
		   \end{array}$
\end{tabular}\vspace{-2ex}
}
\end{center}
\caption{Bounds on the numbers of convex, non-convex, and general 
	\mbox{$k$-gons} for $n$ points and constant $k$.\vspace{2ex}
}
\label{tab:overview_gons}
\end{table*}

%------------------------------ Overview Table Holes --------------------------
\begin{table*}[htb]
\begin{center}
{ \small
\setlength{\tabcolsep}{4pt}
\setlength{\arraycolsep}{1pt}
\renewcommand{\arraystretch}{1.40}
\begin{tabular}{l||l|l|l|l}
		& \multicolumn{1}{c|}{convex} & \multicolumn{1}{c|}{non-convex} & \multicolumn{2}{c}{\hspace{2.9cm}general} \\
		& \multicolumn{1}{c|}{min} & \multicolumn{1}{c|}{max} & \multicolumn{1}{c|}{min} & \multicolumn{1}{c}{max} \\ \hline\hline
%%%
$k\!=\!4$ 	
		& $\begin{array}{ll}  %convex min 4-holes
			%\geq & \binom{n-3}{2} + 6  \\ 
			\geq & \frac{n^2}{2} - \frac{9}{4}n - o(n)  \\
			\leq & 1.9397n^2\!+\!o(n^2) \\
			\multicolumn{2}{l}{\Theta (n^2)~\cite{afhhpv-lbnsc-14, BV2004}} 
		   \end{array}$
		& $\begin{array}{ll}  %non-convex max 4-holes
			\leq &  \frac{n^3}{2}\!-\!O(n^2)  \\
			\geq &  \frac{n^3}{2}\!-\!O(n^2 \log n)  \\
			\multicolumn{2}{l}{\Theta (n^3)~\cite{4HOLES_CGTA}}
		   \end{array}$
		& $\begin{array}{ll}  %general min 4-holes
			\geq & \frac{5}{2}n^2\!-\!O(n)   \\
			\leq & O(n^{\frac{5}{2}}\log n) 	 \\
			\multicolumn{2}{l}{\Omega (n^2)~\cite{4HOLES_CGTA},} \\ %, 
			\multicolumn{2}{l}{O(n^{\frac{5}{2}}\log n)~\mbox{[Sec.~\ref{sec:min_gen_holes}]}}
		   \end{array}$
		& $\begin{array}{l}  %general max 4-holes
			 {n \choose 4}  \\
			 \Theta (n^4)~\cite{4HOLES_CGTA} 
		   \end{array}$								\\\hline
%%%%%%
$k\!=\!5$ 	& $\begin{array}{ll}   %convex min 5-holes
			\geq & \frac{3n}{4} - o(n) \\ 
			\leq & 1.0207n^2\!+\!o(n^2) \\
			\multicolumn{2}{l}{\Omega (n)~\cite{afhhpv-lbnsc-14}, \ O(n^2)~\cite{BV2004}}
		   \end{array}$
		& $\begin{array}{ll}   %non-convex max 5-holes
			\leq &  n! / (n\!-\!4)!  \\
			\multicolumn{2}{l}{\Theta (n^4)}~\mbox{[Sec.~\ref{sec:ub_nc_holes}]} %\\
		   \end{array}$
		& $\begin{array}{ll}  %general min 5-holes
			\geq & 17n^2\!-\!O(n)   \\
			\leq & O(n^3 (\log n)^2) \\
			\multicolumn{2}{l}{\Omega (n^2)~\cite{ahv-5g5h-12}, \ O(n^3(\log n)^2)~\mbox{[Sec.~\ref{sec:min_gen_holes}]}}
		   \end{array}$
		& $\begin{array}{l}  %general max 5-holes
			 {n \choose 5}  \\
			 \Theta (n^5)~\cite{ahv-5g5h-12} 
		   \end{array}$								\\\hline
%%%%%%
$k\!\geq\!6$ &  $\begin{array}{ll} %convex min k-holes
			k\!=\!6\!: & \geq \frac{n}{229}- 4\\
				    & \Omega(n)~\cite{v-ephpp-12} \\ 
				    & O(n^2)~\cite{BV2004} \\  
			k\!\geq\!7\!: & \emptyset~\cite{Ho83} 
		   \end{array}$
		& $\begin{array}{ll} %non-conv max k-holes
			\leq &  n! / (n\!-\!k\!+\!1)!  \\
			 \multicolumn{2}{l}{\Theta (n^{k-1})~\mbox{[Sec.~\ref{sec:ub_nc_holes}]}} %\\
		   \end{array}$
		& $\begin{array}{ll}  %general min k-holes
			\geq & n^2\!-\!O(n)   \\
			\leq & O(n^{\frac{k+1}{2}}(\log n)^{k-3}) \\
			\multicolumn{2}{l}{\Omega (n^2), \ O(n^{\frac{k+1}{2}}(\log n)^{k-3})~\mbox{[Sec.~\ref{sec:min_gen_holes}]}}	%\\
		   \end{array}$
		& $\begin{array}{l}  %general max k-holes
			 {n \choose k}  \\
			 \Theta (n^k)~\mbox{[Sec.~\ref{sec:max_gen_holes}]}
		   \end{array}$	
\end{tabular}%\vspace{-2ex}
}
\end{center}
\caption{Bounds on the numbers of convex, non-convex and general 
	\mbox{$k$-holes} for $n$ points and constant $k$.}
\label{tab:overview_holes}
\end{table*}

%------------------------------ Section 2 --------------------------------
\section{General $k$-gons}
\subsection{$k$-gons and the rectilinear crossing number}
\label{sec:gons}

For small values of $k$, the number of $k$-gons in a point set $S$ of $n$ points 
can be related to the rectilinear crossing number $\rcr(S)$ of $S$.
This is the number of proper intersections 
(i.e., intersections in the interior of edges)
in the (drawing of the) complete straight-line graph on~$S$. By 
$\rcr(n)$ we denote the minimum possible rectilinear crossing number over all
point sets of cardinality $n$. Determining $\rcr(n)$ is a well-known problem in discrete geometry; 
see~\cite{BMP05,EG} as general references and~\cite{web} for bounds on small sets. 
Asymptotically we have $\rcr(n) = c_4 \binom{n}{4} =\Theta(n^4)$, where $c_4$ is a constant
in the range $0.379972 \leq c_4 \leq 0.380473$. The currently best lower
bound on $c_4$ is by {\'A}brego et al.\!~\cite{Abrego2008273, Abrego2012LB2}. 
The upper bound stems from a recent work of Fabila-Monroy and L{\'o}pez~\cite{FaLo14}.

It is easy to see that the number of convex \mbox{4-gons} is equal to
$\rcr(S)$ and is thus minimized by sets realizing $\rcr(n)$.
Since four points in non-convex position span three non-convex
\mbox{4-gons}, we have at most $3 {n \choose 4}\!-\!3 \rcr(n)\approx 1.86 {n \choose 4}$
non-convex and at most $3 {n \choose 4}\!-\!2 \rcr(n)\approx 2.24 {n \choose 4}$ general
\mbox{4-gons}. These bounds are tight for point sets minimizing %that minimize
the rectilinear crossing number.

A similar relation has been obtained for the number of
non-convex \mbox{5-gons} in~\cite{ahv-5g5h-12}: Any set of $n$ points 
has at most 
$10{n \choose 5}-2(n-4)\rcr(n) \approx 6.2{n \choose 5}$ non-convex
5-gons. Again, this bound is achieved by sets minimizing the
rectilinear crossing number. Note that the maximum numbers of non-convex 4- and 5-gons exceed the maximum numbers of their convex counterparts.
For the number of general 5-gons, no such direct relation to $\rcr(n)$ is possible, 
as already for $n=6$ there exist point sets with the same number of crossings but different 
numbers of 5-gons~\cite{ahv-5g5h-12}.
Similarly, for $k \geq 6$, none of the three types of $k$-gons (convex, non-convex, and 
general) in a point set~$S$ can be expressed as a function of $\rcr(S)$.
Still, we can use the rectilinear crossing number to obtain bounds on these numbers. 
Let $g_k^t(S)$ be the number of $k$-gons of \emph{type $t$} 
(\emph{conv}ex, \emph{non-conv}ex, or \emph{gen}eral) in a point set~$S$.

\begin{prop}
\label{prop:kh_gons_rc}
Let $k\geq 4$, and let $c_1$, $c_2$, and $x$ be arbitrary fixed constants such that 
Inequality~(\ref{eqn:kh_gons_start}) holds for all sets $S'$ with cardinality $|S'|=k$.
\begin{equation}
c_1 \leq g_k^t(S') + x \cdot \rcr(S') \leq c_2 \label{eqn:kh_gons_start}
\end{equation}
Then for every point set $S$ with $|S| = n \geq  k$, 
the following bounds hold for the number $g_k^t(S)$ of $k$-gons of type $t$ in $S$.
\begin{eqnarray}
g_k^t(S) & \geq & c_1 \cdot \binom{n}{k} - x \cdot \binom{n-4}{k-4} \cdot \rcr(S) \label{eqn:kh_gons_goal1}\\
g_k^t(S) & \leq & c_2 \cdot \binom{n}{k} - x \cdot \binom{n-4}{k-4} \cdot \rcr(S) \label{eqn:kh_gons_goal2} 
\end{eqnarray}
\end{prop}
\begin{proof}
Given some point set $S$ with $n$ points, consider all its $\binom{n}{k}$ subsets of size $k$ 
$\{S_i \subseteq S: |S_i|=k\}$. Then we have the following equations.
\begin{eqnarray}
	\sum_i \rcr(S_i) & = & \binom{n-4}{k-4} \cdot \rcr(S) \label{eqn:kh_gons_eq1}\\
	\sum_i g_k^t(S_i) & = & g_k^t(S) \label{eqn:kh_gons_eq2}
\end{eqnarray}
Using the first inequality in~(\ref{eqn:kh_gons_start}), Equation~(\ref{eqn:kh_gons_eq2}) can be transformed to the lower bound
\begin{equation*}
g_k^t(S)  \ = \  \sum_i g_k^t(S_i) 
		  \ \geq \  \sum_i \left(c_1 - x \cdot \rcr(S_i)\right) 
		  \ = \  c_1\cdot \binom{n}{k} - x \cdot \sum_i \rcr(S_i) 
\end{equation*}
which, by applying~(\ref{eqn:kh_gons_eq1}),
gives the desired bound~(\ref{eqn:kh_gons_goal1}).
Analogously, we obtain~(\ref{eqn:kh_gons_goal2}) if we combine the second inequality 
in~(\ref{eqn:kh_gons_start}) with Equations~(\ref{eqn:kh_gons_eq1}) and~(\ref{eqn:kh_gons_eq2}).
\end{proof}

If $x$ is negative in Inequality~(\ref{eqn:kh_gons_goal1}), 
then the rectilinear crossing number $\rcr(S)$ of $S$ adds to the lower bound. Thus we 
can replace it by the minimum over all point sets of size $n$, $\rcr(n)$, and by this obtain a lower bound that is 
independent of $S$. 

\begin{cor}
\label{cor:kh_gons_cr_lb}
Assume that for some constants $x \leq 0$ and $c_1$ arbitrary, the inequality $c_1  \leq  g_k^t(S') + x \cdot \rcr(S')$ 
is satisfied for all point sets $S'$ with cardinality $|S'|=k$.
Then for every point set $S$ with $|S| = n \geq  k$ the following lower bound holds for the 
number $g_k^t(S)$ of $k$-gons of type $t$ in $S$.
\begin{equation}
g_k^t(S)  \geq  c_1 \cdot \binom{n}{k} - x \cdot  \binom{n-4}{k-4} \cdot c_4  \binom{n}{4} = 
 \left(c_1 - x \cdot c_4 \cdot \binom{k}{4}\right) \binom{n}{k} \label{eqn:kh_gons_general1}
\end{equation}
\end{cor}

Accordingly, 
if $x$ is positive, we can generalize Inequality~(\ref{eqn:kh_gons_goal2})
to a general upper bound.

\begin{cor}
\label{cor:kh_gons_cr_ub}
Assume that for some constants $x \geq 0$ and $c_2$ arbitrary, the 
inequality $c_2  \geq  g_k^t(S') + x \cdot \rcr(S')$ is satisfied 
for all point sets $S'$ with cardinality $|S'|=k$.
Then for every point set $S$ with $|S| = n \geq  k$, 
the following upper bound applies to the number $g_k^t(S)$ of $k$-gons of type $t$ in $S$.
\begin{equation}
g_k^t(S)  \leq  \left(c_2 - x \cdot c_4 \cdot \binom{k}{4}\right) \binom{n}{k} \label{eqn:kh_gons_general2}
\end{equation}
\end{cor}

	In each of the bounds resulting from Proposition~\ref{prop:kh_gons_rc}, either $c_1$ or $c_2$ is not used. 
	So of course, for independent optimization of the two bounds, it might be helpful to consider the pairs $(c_1,x)$, 
	and $(c_2,x)$ independently, with possibly different optimal values for $x$. 
	On the other hand, optimizing all 
	three values $c_1$, $c_2$, and $x$ simultaneously results in bounds that are more easy to compare.
	In the following we optimize in these two different ways. On the one hand we try to minimize the difference between 
	$c_1$ and $c_2$ to obtain most possibly small ranges for the number of (some type of) $k$-gons in sets with a 
	certain rectilinear crossing number. On the other hand we independently optimize $(c_1,x)$ and $(c_2,x)$ in order to obtain 
	general bounds (meaning bounds that are independent from the rectilinear crossing number of a given set) by applying 
	Corollaries~\ref{cor:kh_gons_cr_lb} and~\ref{cor:kh_gons_cr_ub} .

	Note that concerning the general bounds, lower bounds only make sense for the classical 
	question about convex $k$-gons, as the numbers of general and non-convex $k$-gons are 
	minimized by sets in convex position. 
	Similarly, general upper bounds for non-convex or (general) $k$-gons are
	of interest while the maximum number of convex $k$-gons is $\binom{n}{k}$, 
	again achieved by convex sets.

	Recall that the number of $k$-gons (of whatever type) in a point set $S$ 
	only depends on the combinatorial properties and thus on the order type $OT(S)$ of the 
	underlying point set $S$.
	Thus, we calculate pairs $(g_k^t(OT),\rcr(OT))$ for all possible order types $OT$
	of $k$ points, this way obtaining all possible pairs 
	$(g_k^t(S'),\rcr(S'))$ that can occur for any point set $S'$ with  $|S'|=k$.
	For the calculation, we use the order type database~\cite{AK01}, that contains 
	a complete list of the order types of up to $11$ points. 
	Having this, we can optimize the values $c_1$, $c_2$, and $x$ that 
	fulfill~(\ref{eqn:kh_gons_start}), 
	obtaining the according relations~(\ref{eqn:kh_gons_goal1}) and~(\ref{eqn:kh_gons_goal2}).
	Tables~\ref{tab:ordertype_relations_relative} and~\ref{tab:ordertype_relations} 
	show an overview of the resulting relations.
	Note that trivial bounds on the numbers of $k$-gons can be obtained by assuming the theoretic 
	possible maximum number of $k$-gons for each $k$-tuple. 
	As the maximum numbers of general / non-convex $k$-gons in a $k$-tuple are 8, 29, and 92 
	for $k \in \{5,6,7\}$, the bounds in Table~\ref{tab:ordertype_relations} all improve over these 
	trivial bounds.

	\begin{table}[htb]
	\begin{center}
	{\small
	\renewcommand{\arraystretch}{1.45}
	\setlength{\tabcolsep}{3pt}
	\begin{tabular} {rclcl}
	\hline
	\vspace{-10pt}\\
	$-0.75 \binom{n}{5} + 0.25 \cdot (n-4) \cdot \rcr(S)$ & $\leq$ & 
		$g_5^{conv}(S)$ & $\leq$ & 
		$-0.25 \binom{n}{5} + 0.25 \cdot (n-4) \cdot \rcr(S)$\\
	 &  & 
		$g_5^{non-conv}(S)$ & $=$ & $10 \binom{n}{5} - 2 \cdot (n-4) \cdot \rcr(S)$ \\
	$9.25 \binom{n}{5} - 1.75 \cdot (n-4) \cdot \rcr(S)$ & $\leq$ & 
		$g_5^{gen}(S)$ & $\leq$ & 
		$9.75 \binom{n}{5} - 1.75 \cdot (n-4) \cdot \rcr(S)$\\
	$- \binom{n}{6} + 0.08\dot{3} \binom{n-4}{2} \cdot \rcr(S)$ & $\leq$ & 
		$g_6^{conv}(S)$ & $\leq$ & 
		$-0.25 \binom{n}{6} + 0.08\dot{3} \binom{n-4}{2} \cdot \rcr(S)$\\
	$29\frac{4}{9} \binom{n}{6} - \frac{22}{9} \binom{n-4}{2} \cdot \rcr(S)$ & $\leq$ & 
		$g_6^{non-conv}(S)$ & $\leq$ & 
		$36\frac{6}{9} \binom{n}{6} - \frac{22}{9} \binom{n-4}{2} \cdot \rcr(S)$\\
	$28\frac{1}{3} \binom{n}{6} - \frac{7}{3} \binom{n-4}{2} \cdot \rcr(S)$ & $\leq$ & 
		$g_6^{gen}(S)$ & $\leq$ & 
		$36 \binom{n}{6} - \frac{7}{3} \binom{n-4}{2} \cdot \rcr(S)$\\
	$-1.1\overline{923076} \binom{n}{7} + 0.0\overline{384615} \binom{n-4}{3} \cdot \rcr(S)$ & $\leq$ & 
		$g_7^{conv}(S)$ & $\leq$ & 
		$-0.3\overline{461538} \binom{n}{7} + 0.0\overline{384615} \binom{n-4}{3} \cdot \rcr(S)$\\
	$86.\overline{230769} \binom{n}{7} - 3.\overline{538461} \binom{n-4}{3}\cdot \rcr(S)$ & $\leq$ & 
		$g_7^{non-conv}(S)$ & $\leq$ & 
		$123.\overline{846153} \binom{n}{7} - 3.\overline{538461} \binom{n-4}{3} \cdot \rcr(S)$\\
	$85.5 \binom{n}{7} - 3.5 \binom{n-4}{3} \cdot \rcr(S)$ & $\leq$ & 
		$g_7^{gen}(S)$ & $\leq$ & 
		$123.5 \binom{n}{7} - 3.5 \binom{n-4}{3} \cdot \rcr(S)$\\
	\vspace{-10pt}\\
	\hline
	\end{tabular}\vspace{-2ex}
	}
	\end{center}
	\caption{Bounding the number of $k$-gons in an $n$ point set $S$ via its rectilinear crossing number $\rcr(S)$.}
	\label{tab:ordertype_relations_relative}
	\end{table}

\begin{table}[htb]
\begin{center}
{\small
\renewcommand{\arraystretch}{1.40}
\setlength{\tabcolsep}{3pt}
\begin{tabular} {lcrrr}
\hline
\vspace{-10pt}\\
$g_5^{non-conv}(S)$ & $\leq$ & $(10-10 \cdot c_4)\binom{n}{5}$ & $\approx$ & $ 6.20\binom{n}{5}$ \\
$g_5^{gen}(S)$ & $\leq$ & $(9.75-8.75 \cdot c_4)\binom{n}{5}$ & $\approx$ & $ 6.43\binom{n}{5}$ \\
$g_6^{non-conv}(S), \ g_6^{gen}(S)$ & $\leq$ & $(36 - 35 \cdot c_4)\binom{n}{6}$ & $\approx $ & $22.7\binom{n}{6}$ \\
$g_7^{non-conv}(S)$ & $\leq$ & $(123.\overline{846153} - 123.\overline{846153} %3.\overline{538461} \cdot 35 
    \cdot c_4)\binom{n}{7}$ & $\approx $ & $75.64\binom{n}{7}$ \\
$g_7^{gen}(S)$ & $\leq$ & $(123.5 - 122.5 \cdot c_4)\binom{n}{7}$ & $\approx $ & $76.95\binom{n}{7}$ \\
\vspace{-10pt}\\
\hline
\end{tabular}\vspace{-2ex}
}
\end{center}
\caption{Bounding the number of $k$-gons in an $n$ point set $S$ via the constant $c_4$ of the (minimum) rectilinear crossing number $\rcr(n)$, $\rcr(n) = c_4\binom{n}{4}$.}
\label{tab:ordertype_relations}
\end{table}

From the calculations it can be seen that for all sets of size $k \leq 7$, 
the point sets reaching the maximum number of general or 
non-convex $k$-gons are at the same time minimizing the number of crossings. 
The same is true for $k=8$.
But continuing the calculations until $k=9$, it turns out that this is not true in general. 
The (combinatorially unique) point set containing the maximum number of 1282 general 9-gons has 38 crossings and 
thus does not reach the (minimum) rectilinear crossing number $\rcr(9)=36$~\cite{web}.
%%%%%%%%%%%%%%%%%%%%%%%%%%%%%%%%%%%%%%%%%%%%%%%%%%%%%%%%%%%%%%%%%%%

%\newpage
\subsection{$k$-gons, polygonizations, and the double chain\index{polygonization|(}\index{double chain}}

Polygonizations, also called spanning cycles, can be considered 
as $k$-gons of maximal size (i.e., $k=n$). 
Garc\'ia et al.\!~\cite{GNT00} construct a point set with
$\Omega(4.64^n)$ spanning cycles, the so-called double chain $DC(n)$, 
which is currently the best known minimizing example; see Figure~\ref{fig:dchain_orig}.

\begin{figure}[htb]
  \centering
  \includegraphics[scale=0.5]{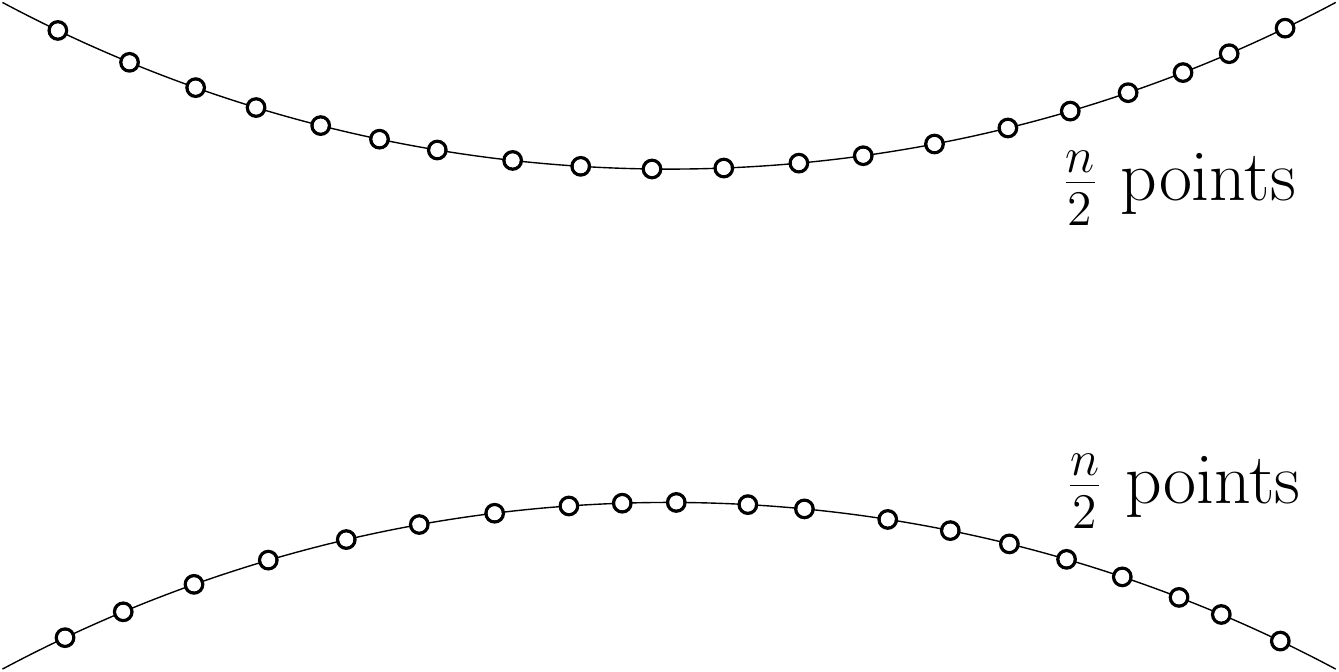} 
  \caption{The so-called double chain $DC(n)$.}
  \label{fig:dchain_orig}
\end{figure}

The upper bound on the number of spanning cycles of any $n$-point set 
was improved several times during the last years, most recently to
$O(68.664^n)$~\cite{dsst-bmmsc-11} and 
$O(54.543^n)$~\cite{ssw-gpg-13}, neglecting polynomial factors in
the asymptotic expressions. The minimum is achieved by point sets 
in convex position, which have exactly one spanning cycle. 
For the number of general $k$-gons this implies 
a lower bound of ${n \choose k}$, as well as an 
upper bound of $O\left((54.543^k{n \choose k}\right)$. 
Hence, for constant $k$, any point set has $\Theta(n^k)$ general $k$-gons.
On the other hand, the double chain
provides $\Omega(n^k)$ non-convex \mbox{$k$-gons}, 
where $k \geq 4$ is again a constant. To see this, 
choose one vertex from the upper chain of
$DC(n)$ and $k-1 \geq 3$ vertices from the lower chain of $DC(n)$, and connect
them to a simple, non-convex polygon. This gives at least $\frac{n}{2} {n/2
  \choose k-1} = \Omega(n^k)$ non-convex $k$-gons. As the lower bound
on the maximal number of non-convex $k$-gons asymptotically matches 
the upper bound on the maximal number of general $k$-gons, we obtain
the following result.

\begin{prop}
\label{prop:kgons}
For any constant $k \geq 4$, the number of non-convex \mbox{$k$-gons} in a set of
$n$ points is bounded by $O(n^k)$. This is tight in the sense that there exist sets with
$\Omega(n^k)$ non-convex \mbox{$k$-gons}.
\end{prop}

%%%%% k-holes %%%%%%%%%%%%%%%%%%%%%%%%%%%%%%%%%%%%%%%%%%%%%%%%%%%%%%%%%%%%%%%%%%%%%
\section{Maximizing the number of (general) \mbox{$k$-holes}}
\label{sec:max_gen_holes}

In~\cite{4HOLES_CGTA} it is shown that the number of \mbox{4-holes} is maximized for 
point sets in convex position if $n$ is sufficiently large. It was
conjectured that this is true for any constant $k \geq 4$.
The following theorem settles this conjecture in the affirmative. %to the positive.

\begin{theorem}
\label{thm:convex_k}
  For every $k\geq 4$ and $n \geq 2(k-1)! \binom{k}{4}\!+\!k\!-\!1$, the number of 
  \mbox{$k$-holes} is maximized by a set of $n$ points in convex position.
\end{theorem}
\begin{proof}
  Consider a non-convex \mbox{$k$-hole} $H$. 
  For each of its non-extreme vertices (i.e., vertices not on the convex hull of $H$), 
  there exists a triangle spanned by three extreme vertices of $H$ such that the 
  non-extreme vertex is contained in the interior of the triangle. 
  Further, there exists at least one reflex (and thus non-extreme) vertex $v_r$ 
  of $H$ such that removing $v_r$ from the vertex set of $H$ 
  (and connecting its incident vertices) results in a simple non-empty \mbox{$(k\!-\!1)$-gon}. 
  To see the latter, consider an edge $e$ of $\CH(H)$ which is not in the boundary of $H$. 
  Together with some part of the boundary of $H$, $e$ forms a simple polygon $H'$ that is interior-disjoint with $H$. 
  For the case where $H'$ is just a triangle, $e$ can be used to cut off the third vertex of $H'$ from $H$. 
  Clearly, the resulting \mbox{$(k\!-\!1)$-gon} is simple.
  If $H'$ has at least four vertices then any triangulation of $H'$ has at least two ears.
  For any ear not incident to $e$, the according diagonal of the triangulation can be used to cut off the central vertex of the ear from $H$. 
  Again, the resulting \mbox{$(k\!-\!1)$-gon} is simple.

  Now consider a non-empty triangle $\Delta$. We give an upper bound for the number of non-convex $k$-holes having the three vertices of $\Delta$ as extreme points.
  Denote by $\mathcal{K}$ the set of simple non-empty \mbox{$(k\!-\!1)$-gons} having the vertices of 
  $\Delta$ on their convex hull. 
  First, $|\mathcal{K}|$  can be bounded from above by the number of simple, 
  possibly empty \mbox{$(k\!-\!1)$-gons} having the three vertices of 
  $\Delta$ on their boundary, that is, $|\mathcal{K}| \leq \frac{(k-2)!}{2} \binom{n-3}{k-4}$.

  Further, every simple \mbox{$(k\!-\!1)$-gon} in $\mathcal{K}$ may be completed to a 
  simple non-convex \mbox{$k$-hole} in at most $k\!-\!1$ ways  by adding a reflex vertex: 
  As the resulting polygon has to be empty, we have to use the inner geodesic connecting the two adjacent vertices of the \mbox{$(k\!-\!1)$-gon}.
  Only if this geodesic contains exactly one point, we do obtain one non-convex \mbox{$k$-hole}.
  Thus the number of non-convex \mbox{$k$-holes} having all vertices of $\Delta$ on their convex hull is bounded 
  from above by 
  \begin{equation*}
  (k-1) \frac{(k-2)!}{2} \binom{n-3}{k-4} = \frac{(k-1)!}{2} \binom{n-3}{k-4}. 
  \end{equation*}

  Considering convex \mbox{$k$-holes}, observe that every $k$-tuple gives at most one convex \mbox{$k$-hole}. 
  Denote by $N$ the number of $k$-tuples that do \emph{not} form a convex \mbox{$k$-hole}, 
  and by $T$ the number of non-empty triangles. Then we get~(\ref{eqn:k_bound1})
  as a first upper bound on the number of (general) \mbox{$k$-holes} of a point set.
  \begin{equation}
  \label{eqn:k_bound1}
  \binom{n}{k} - N + \left(\frac{(k-1)!}{2} \binom{n-3}{k-4}\right) \cdot T 
  \end{equation}

  To obtain an improved upper bound from~(\ref{eqn:k_bound1}), 
  we need to derive a good lower bound for $N$. To this end, 
  consider again a non-empty triangle $\Delta$. 
  As $\Delta$ is not empty, 
  none of the $\binom{n-3}{k-3}$ \mbox{$k$-tuples} that contain all three 
  vertices of $\Delta$ forms a convex \mbox{$k$-hole}.
  On the other hand, for such a \mbox{$k$-tuple}, all of its $\binom{k}{3}$ contained triangles
  might be non-empty. Thus, we obtain $T\cdot\binom{n-3}{k-3} / \binom{k}{3}$ as a lower bound for $N$ 
  and~(\ref{eqn:k_bound2}) as an upper bound for the number of \mbox{$k$-holes}.
  \begin{equation}
  \label{eqn:k_bound2}
  \binom{n}{k} + \left(\frac{(k-1)!}{2} \binom{n-3}{k-4} - \frac{\binom{n-3}{k-3}}{\binom{k}{3}}\right) \cdot T 
  \end{equation}
  
  For $n \geq 2(k-1)! \binom{k}{4} + k - 1$ this is at most
  $\binom{n}{k}$, the number of \mbox{$k$-holes} of a set of $n$~points in
  convex position, which proves the theorem.
\end{proof}

The above 
theorem states that convexity maximizes the number of 
\mbox{$k$-holes} for $k = O(\frac{\log n}{\log\log n})$ and sufficiently 
large $n$. Moreover, the
proof implies that any non-empty triangle in fact reduces 
the number of empty \mbox{$k$-holes}. 
Thus it follows that, for $k = O(\frac{\log n}{\log \log n})$ 
and $n$ sufficiently large, the 
maximum number of convex \mbox{$k$-holes} is strictly larger than 
the maximum number of non-convex \mbox{$k$-holes}; see also the next section.

At the other extreme, 
for $k\!\approx \! n$ the statement does not hold: 
As already mentioned in the introduction, 
a set of $k$ points spans at most one convex \mbox{$k$-gon},
but might admit exponentially many different non-convex $k$-gons~\cite{GNT00}.
This leads to the question, for which $k$ the situation changes.
The following theorem implies that for some $0<c<1$ and every $k \geq c\cdot n$, 
the convex set does not maximize the number of $k$-holes.

\begin{theorem}
\label{thm:dchain} 
The number of \mbox{$k$-holes} in the double chain $DC(n)$ on $n$ points is at least 
\[\binom{\frac{n-4}{2}}{\frac{n-k}{2}} \cdot  \frac{n-k+2}{2} \cdot \Omega(4.64^k). \]
\end{theorem}
\begin{proof}
Recall that %As already mentioned in Section~\ref{sec:gons},
Garc\'ia et al.\!~\cite{GNT00} showed that the double 
chain on $n$ points ($n/2$ points on each chain)
admits $\Omega (4.64^n)$ polygonizations. 
To estimate the number of \mbox{$k$-holes} of the double chain on $n$ points, we first use this result 
for a double chain on $k$ points ($k/2$ points on each chain), 
obtaining $\Omega (4.64^k)$ different \mbox{$k$-polygonizations}. Then we
distribute the remaining $n-k$ points among all possible positions, 
meaning that for each $k$-polygonization, we obtain 
the double chain on $n$~points with 
a \mbox{$k$-hole} drawn; see Figure~\ref{fig:dchain}.

In their proof, Garc\'ia et al.\! count paths that start at the first vertex of the upper chain and end 
at the last vertex of the lower chain. Before the first 
vertex on the lower chain, they add an additional point $q$ 
to complete these paths to polygonizations. 
We slightly extend this principle, by also adding an additional point $p$ on the upper chain after the last 
vertex; see Figure~\ref{fig:dchain_path}.
\begin{figure}[htb]
  \centering
  \includegraphics[scale=0.55]{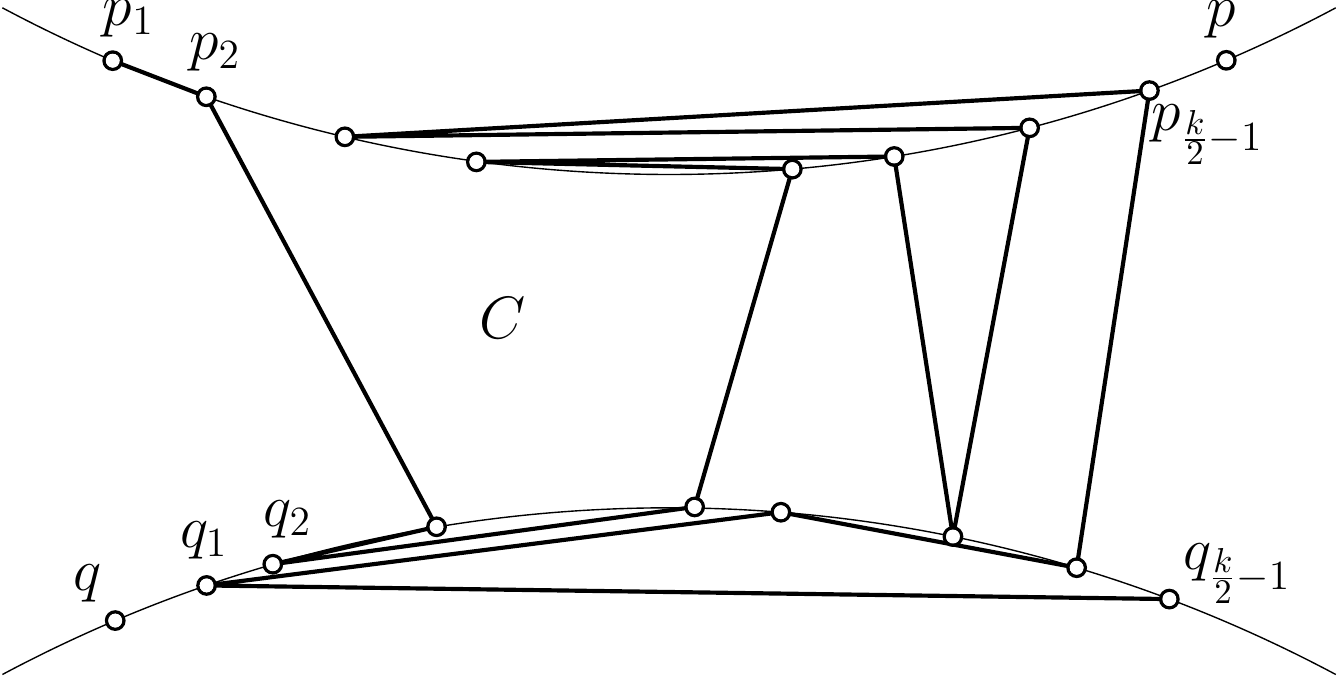} 
  \caption{A path $C$ in the double chain, using all but the vertices $p$ and~$q$. }
  \label{fig:dchain_path}
%\vspace{-2ex}
\end{figure}

Then we complete each path 
$C$ to a polygonization in one of the following ways: 
Either we add $p$ to $C$ directly next to $p_{\frac{k}{2}-1}$ and then complete $C$ via $q$, 
obtaining $P_q$, or we add $q$ to $C$ directly 
next to $q_1$, and close the polygonization via $p$, obtaining $P_p$; see again Figure~\ref{fig:dchain}.

\begin{figure}[htb]
  \centering
  \includegraphics[scale=0.55]{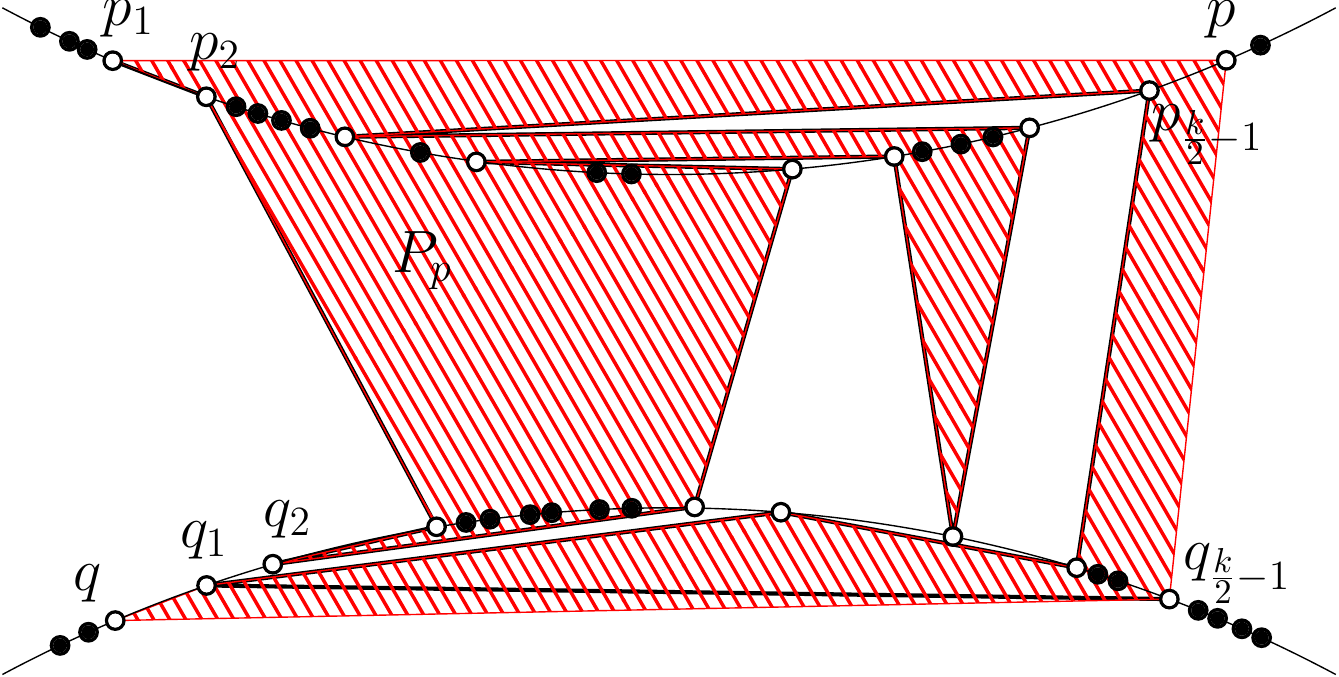} \hspace{0.6cm} 
  \includegraphics[scale=0.55]{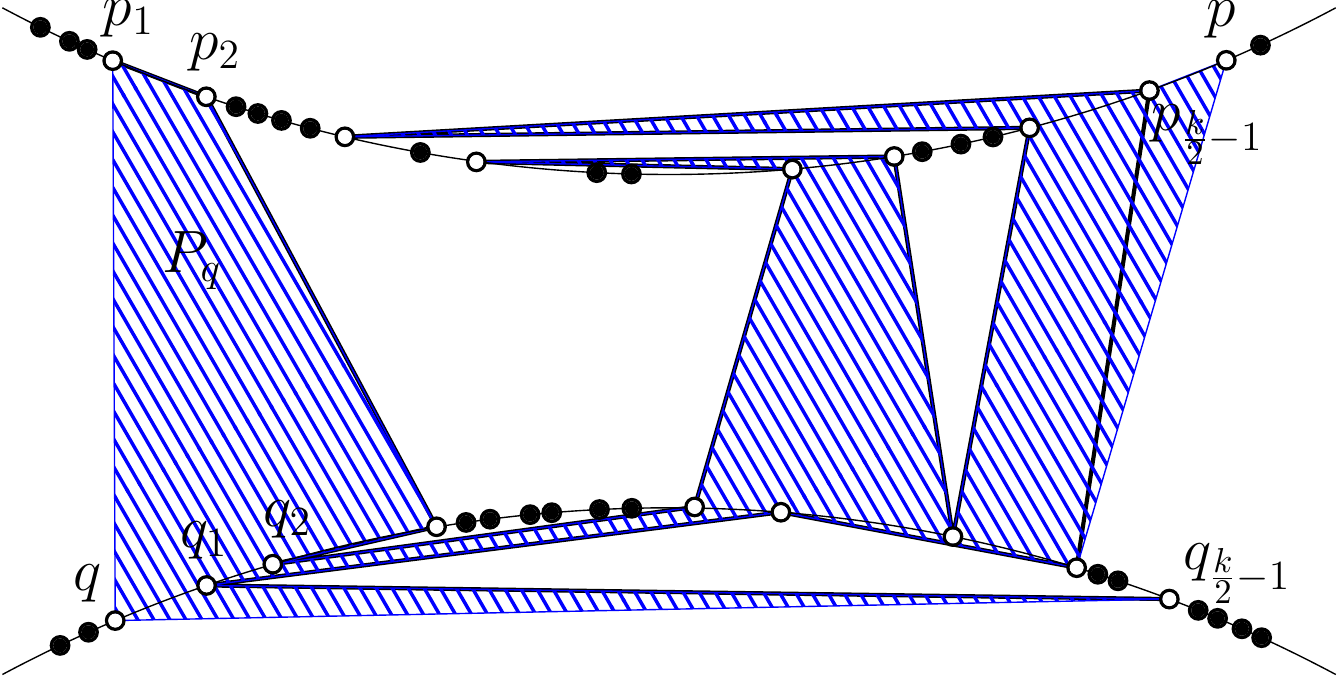} 
  \caption{Two ways to complete a path to a polygonization. }
  \label{fig:dchain}
%\vspace{-1ex}
\end{figure}

Note that this changes the number of 
polygonizations only by a constant factor 
and thus does not influence the asymptotic bound. 
However, the interior of $P_q$ is 
the exterior of its  ``complemented'' polygonization 
$P_p$, meaning that if we place a point somewhere on the double chain and it lies 
inside $P_q$, then it lies outside $P_p$, and vice versa. It follows that, in one of the two polygonizations, at least 
half of the $k+2$ positions to insert points are %``valid'' (
outside the polygonization. %).
Hence, we can distribute the $\frac{n-k}{2}$ points on each chain 
to at least $\frac{k}{2}+1$ possible positions in total.
Now, on one of the two chains we have at least $\frac{k}{4} + 1$ positions; 
see again Figure~\ref{fig:dchain}.
More precisely, there are $\frac{k}{4}+j+1$ positions on this chain (where $0 \leq j < \frac{k}{4}$)
and  (at least)  $\max\{2,\frac{k}{4}-j\}$ positions on the other chain.
The lower bound stems from the fact that the positions before the first 
and after the last vertex of a chain are always possible.
Placing $a$ points on the $b$ positions of one chain can be seen as placing $a$ balls into $b$ boxes.
The number of ways to do so is $\binom{a+b-1}{a}$.
Using this, we obtain   
  \[\binom{\frac{n-k}{2}\!+\!\frac{k}{4}\!+\!j}{\frac{n-k}{2}} \cdot 
     \max\left\{ \binom{\frac{n-k}{2}\!+\!1}{\frac{n-k}{2}}, \binom{\frac{n-k}{2}\!+\!\frac{k}{4}\!-\!j\!-\!1}{\frac{n-k}{2}} \right\} 
  \] 
  possibilities to place the remaining points on the two chains. This factor is 
  minimized for $j = \frac{k}{4}-2$, which yields the claimed lower bound of
 \[\binom{\frac{n-4}{2}}{\frac{n-k}{2}} \cdot  \frac{n-k+2}{2} \cdot \Omega(4.64^k) \]
 for the number of \mbox{$k$-holes} of $DC(n)$.
\end{proof}

\section{An upper bound for the number of non-convex \mbox{$k$-holes}}
\label{sec:ub_nc_holes}
The following theorem shows that for sufficiently small $k$ with respect to $n$,
the maximum number of non-convex \mbox{$k$-holes} 
is smaller than the maximum number of convex \mbox{$k$-holes}.

\begin{theorem}
\label{thm:max_nc_holes}
  For any constant $k\geq 4$, the number of non-convex \mbox{$k$-holes} in a set of
  $n$ points is bounded by $O(n^{k-1})$ and there exist sets with
  $\Omega(n^{k-1})$ non-convex \mbox{$k$-holes}.
\end{theorem}
\begin{proof}
  We first show that there are at most $O(n^{k-1})$ non-convex
  \mbox{$k$-holes}
  by giving an algorithmic approach to generate all non-convex \mbox{$k$-holes}. 
  We represent a non-convex \mbox{$k$-hole} by the
  counter-clockwise sequence of its vertices, where we require that
  the last vertex is reflex and its removal results in a simple $(k\!-\!1)$-gon; 
  see again the proof of Theorem~\ref{thm:convex_k}. 
  Any non-convex \mbox{$k$-hole} has $r \geq 1$ 
  such representations, where $r$ is at most the number of its reflex
  vertices. Thus the number of different representations is an upper
  bound on the number of non-convex \mbox{$k$-holes}.

  For $1\leq i \leq k\!-\!1$, we have $n\!-i\!+\!1$ possibilities to choose the $i$-th vertex $v_i$. 
  If the resulting $(k\!-\!1)$-gon is non-simple, we ignore it (but still count it).
  For the last vertex $v_k$, we have at most one possibility:
  As $v_k$ is required to be reflex and the polygon has to be
  empty, we have to use the inner geodesic connecting $v_{k-1}$ back
  to $v_1$. Only if this geodesic contains exactly one point,
  namely~$v_k$, we do obtain one non-convex \mbox{$k$-hole}.
  Altogether, we obtain at most
  $n(n-1)(n-2)\ldots(n-k+2)=n!/(n\!-\!k\!+\!1)! = O(n^{k-1})$ non-convex \mbox{$k$-holes}.

\begin{figure}[ht]
  \centering
  \includegraphics[scale=0.5]{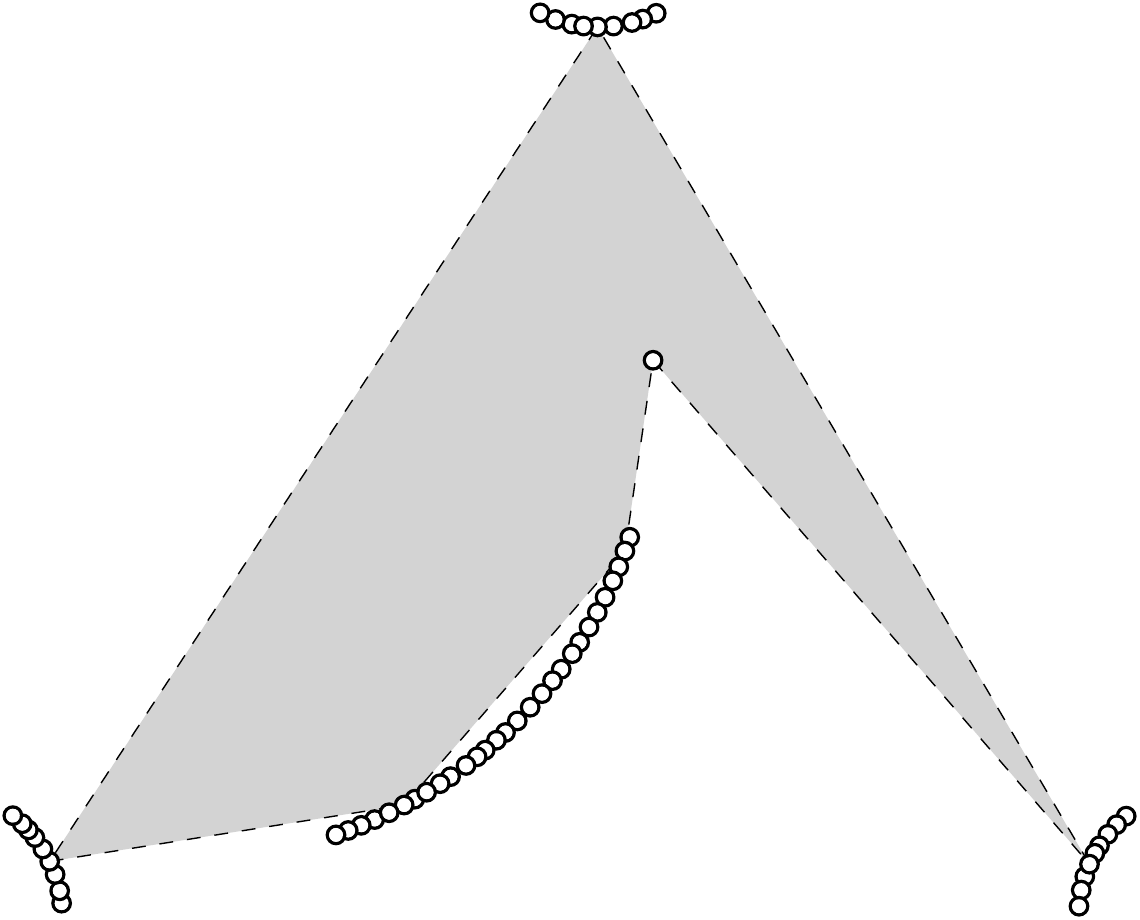}
  \caption{A set with $\Theta(n^{k-1})$ non-convex \mbox{$k$-holes}.}
  \label{fig:nonconv_k_holes}
\end{figure}

  For an example achieving this bound see Figure~\ref{fig:nonconv_k_holes}. 
  Each of the four indicated groups of points contains a linear fraction of the point set; for example $\frac{n}{4}$ points. 
  To show that in this example we have $\Omega(n^{k-1})$
  non-convex \mbox{$k$-holes} it is sufficient to only consider the \mbox{$k$-holes}
  with triangular convex hull of the type indicated in the figure.
  For each of the three vertices of the convex hull of the \mbox{$k$-hole} we
  have a linear number of possible choices, and the $k-4$ non-reflex
  inner vertices can also be chosen from a linear number of
  vertices. 
  Hence, we obtain 
  $\Omega\left(n^3 \cdot \binom{n}{k-4}\right) = \Omega(n^{k-1})$
  non-convex \mbox{$k$-holes}.
\end{proof}

\section{On the minimum number of (general) \mbox{$k$-holes}}
\label{sec:min_gen_holes}

We start with an upper bound on the minimum (over all $n$-point sets) number of (general) \mbox{$k$-holes}.
Note that the minimum number of (general) \mbox{$k$-holes} cannot be greater than 
the minimum number of convex \mbox{$k$-holes} 
plus the maximum number of non-convex \mbox{$k$-holes}.
Recall that the minimum number of convex \mbox{$k$-holes} is $O(n^2)$ for $k \leq 6$ (and zero for $k \geq 7$), and 
the maximum number of non-convex \mbox{$k$-holes} is $O(n^{k-1})$. 
For $k \geq 3$, the latter dominates the former, yielding an upper bound of $O(n^{k-1})$
for the minimum number of general \mbox{$k$-holes}. 
But this bound is by far not tight, as the following considerations show.\\[-1ex]

\subsection{An upper bound on the minimum number of (general) \mbox{$k$-holes} using grids}
\label{sec:min_ub_horton}

Consider an integer grid $G$ of size $\sqrt{n} \times \sqrt{n}$. 
We denote a segment spanned by two points of~$G$ that does not have any points of~$G$ 
in its interior as \emph{prime segment}. 
Further, we denote the \emph{slope} of a line~$l$ spanned by points of $G$ as the differences $(d_x,d_y)$ 
of the coordinates of the endpoints of a prime segment on $l$.
Note that a line with slope $(0,1)$ or $(1,0)$ contains exactly $\sqrt{n}$ points 
of~$G$. A line with slope $(d_x,d_y)$, both $d_x,d_y\neq 0$,  
contains at most $\min \left\{\LC\frac{\sqrt{n}}{|d_x|}\RC,\LC\frac{\sqrt{n}}{|d_y|}\RC \right\}$ 
points of~$G$.
A $k$-gon spanned by points of $G$ is called \emph{interior-empty} if it does not contain any points of $G$ in its interior.

\begin{lemma}
\label{lem:kh_ub_grid}
In an integer grid $G$ of size $\sqrt{n} \times \sqrt{n}$, 
every segment is incident to at most $O(\sqrt{n} \log n)$ interior-empty triangles.
\end{lemma}

\begin{proof}
Consider an arbitrary segment $pq$ spanned by points of~$G$ 
(and possibly containing some points of~$G$ in its interior),
and its supporting line $l$. 
Let $l'$ and $l''$ be the two lines parallel to $l$ and going through 
points of~$G$ such that no point of~$G$ lies between $l$ and $l'$,
and between $l$ and $l''$, respectively; see Figure~\ref{fig:kh_reg_grid} 
where the segment $pq$ is a prime segment.

\begin{figure}[ht]
  \centering
  \includegraphics[page=4, scale=1]{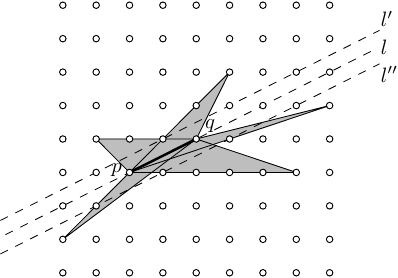} 
  \caption{A prime segment $pq$ in a $9\times 9$ integer grid with the according lines $l$, $l'$, and $l''$.
			Gray triangles are interior-empty and have the third point outside the strip $l'l''$.}
  \label{fig:kh_reg_grid}
\end{figure}

Both, $l'$ and $l''$, 
contain at most $\sqrt{n}$ points of~$G$, each spanning %of which spans 
an interior-empty
triangle with $pq$. Further, each of the points
on $l$ spans a degenerate interior-empty triangle with $pq$. 
Any other triangle $\Delta$ incident to $pq$ has its third point $r$ 
strictly outside the strip bounded by $l'$ and~$l''$.

A necessary condition for such a triangle $\Delta$ to be interior-empty is that both, $pr$ and~$qr$,
`pass through' the same prime segment $s$ on $l'$ (or $l''$). 
Moreover, at least one of the supporting lines of~$pr$ and $qr$ 
contains an endpoint of this prime segment~$s$. 
To see this, w.l.o.g., let $s$ be a prime segment on $l'$. 
Further, let $l_p$ be the line through~$p$ and one endpoint of~$s$ and $l_q$ be the line through $q$ and the other endpoint of~$s$ such that $l_p$ and $l_q$ do not cross between $l$ and $l'$. 
See Figure~\ref{fig:kh_reg_grid_lp_lq}.

\begin{figure}[htb]
  \centering
	  \includegraphics[page=6, scale=1]{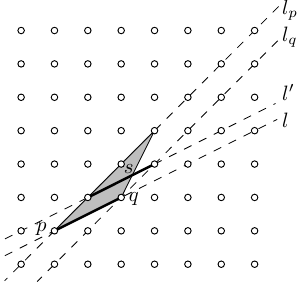} \hspace{2cm}
	  \includegraphics[page=7, scale=1]{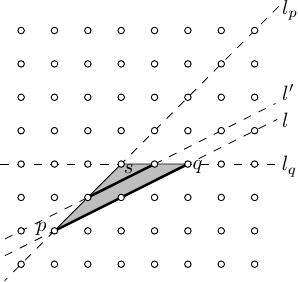} 
  \caption{Lines $l_p$ and $l_q$ for a non-prime segment $pq$ (left) or a prime segment $pq$ (right).}
  \label{fig:kh_reg_grid_lp_lq}
\end{figure}

Consider the region $A$ that contains $s$ and is bounded by $l$, $l_p$, and $l_q$. 
Note that $A$ is a half-strip, if $pq$ is a prime segment, or triangular, otherwise; see Figure~\ref{fig:kh_reg_grid_lp_lq} again.
Moreover, as $s$ is a prime segment, the interior of $A$ does not contain any point of~$G$.
Thus, any point seen from both $p$ and $q$ through $s$ lies on the boundary of $A$, more precisely,  on $l_p$ or $l_q$.

Using this latter property, we can derive an upper bound on 
the number of points $r$ that are visible from $p$ and $q$ 
via the same prime segment~$s$ by counting the number of 
points on such supporting lines. 

\enlargethispage{3ex}
Consider first the case that $pq$ is a horizontal segment, i.e., $q-p=(d_x,0)$. 
Then the according lines through $p$ (or $q$) 
and a point of a prime segment on $l'$ (or $l''$) have slopes in the range
$\{(0,1), (\pm1,1), (\pm2,1), \ldots, (\pm\sqrt{n},1)\}$.
Assuming that all these lines in fact exist for $pq$ in $G$,
we obtain the following  upper bound 
for the total number of interior-empty triangles incident to $pq$
(the first term $3\sqrt{n}$ is for triangles having the third point on one of $l$, $l'$, and~$l''$).
\begin{equation}
\label{eqn:grid}
\begin{array}{rcl}
\displaystyle 3\sqrt{n} + 2\cdot\sum_{i=-\sqrt{n}}^{-1} \LC\frac{\sqrt{n}}{|i|}\RC + 2\cdot\sqrt{n}
+ 2\cdot\sum_{i=1}^{\sqrt{n}} \LC\frac{\sqrt{n}}{|i|}\RC
& = &    \displaystyle 	 O(\sqrt{n}\log(n))
\end{array}
\end{equation}

For the general case of $pq$ being a segment with $q-p=(d_x,d_y)$,
its supporting line $l$ has slope 
$(d'_x,d'_y)=(\frac{d_x}{\gcd(d_x,d_y)},\frac{d_y}{\gcd(d_x,d_y)})$.
The according slopes of the lines through $p$ (or $q$) 
and a point of a prime segment on $l'$ (or $l''$)
differ from each other by a multiple of $(d'_x,d'_y)$. 
Note that $d'_x$ and $d'_y$ are integers with $\max\{d'_x,d'_y\} \geq 1$. Thus, 
the according number of interior-empty triangles for a general segment cannot exceed 
the bound in~(\ref{eqn:grid}) for a horizontal segment, which completes the proof.
\end{proof}

\begin{theorem}
\label{thm:ub_min_gen}
For every constant $k \geq 4$ and every number \mbox{$n \geq k$} with $\sqrt{n} \in \mathbb{N}$,
 there exist sets with $n$ points in general 
position that admit at most $O(n^{\frac{k+1}{2}}(\log n)^{k-3})$ \mbox{$k$-holes}.
\end{theorem}
\begin{proof}
The point set $S$ we consider is the squared Horton set 
of size $\sqrt{n} \times \sqrt{n}$; see~\cite{Va92}.
Roughly speaking, $S$ is a grid which is perturbed such that 
every set of originally collinear points forms a Horton set.
For any two points $p,q \in S$, the number of empty triangles 
that have the segment $pq$ as an edge is bound from above 
by the number of (possibly degenerate) interior-empty triangles 
incident to the according segment in the regular grid.
By Lemma~\ref{lem:kh_ub_grid}, this latter number is at most $O(\sqrt{n}\log n)$.

To estimate the number of \mbox{$k$-holes} in $S$, we use triangulations and their dual: 
In the dual graph of a triangulation, every triangle is 
represented as a node, and two nodes are connected iff the corresponding triangles share an edge. 
For the triangulation of a \mbox{$k$-hole}, this gives a %graph is a 
binary tree which can be rooted at any triangle that has an edge on the boundary of the \mbox{$k$-hole}.
It is well-known that there are $O(4^k \cdot k^{-\frac{3}{2}})$ such rooted binary trees~\cite{TR10}.
Although exponential in $k$, this bound is constant 
with respect to the size $n$ of~$S$.

Now pick an empty triangle $\Delta$ in $S$ and an arbitrary rooted binary tree~$B$ of size $k\!-\!2$. 
Consider all \mbox{$k$-holes} which %contain $\Delta$ and 
admit a triangulation whose dual is $B$~rooted at~$\Delta$. 
Note that each such $k$-hole contains~$\Delta$.
We ``build'' and count (triangulations of) these $k$-holes by starting from $\Delta$ and following the edges of~$B$.
As by Lemma~\ref{lem:kh_ub_grid}, the number of empty triangles incident to an edge in~$S$ is $O(\sqrt{n}\log n)$, 
each of the $k\!-\!3$ edges in~$B$ yields $O(\sqrt{n}\log n)$ 
possibilities to continue a triangulated \mbox{$k$-hole}, 
and we obtain an upper bound of 
$O((\sqrt{n}\log n)^{k-3})$ 
for the number of triangulations of \mbox{$k$-holes} for~$\Delta$ that represent~$B$.

Multiplying this by the (constant) number of rooted binary trees of size $k-2$ 
does not change the asymptotics and thus yields an upper bound of 
$O((\sqrt{n}\log n)^{k-3})$ 
for the number of all triangulations of all \mbox{$k$-holes} containing~$\Delta$.
As any \mbox{$k$-hole} can be triangulated, this is also an upper bound for the 
number of \mbox{$k$-holes} containing~$\Delta$.

Finally, as there are  $O(n^2)$ empty triangles in $S$ (see~\cite{BV2004}),
we obtain  $O(n^2(\sqrt{n}\log n)^{k-3}) =  \linebreak O(n^{\frac{k+1}{2}}(\log n)^{k-3})$ 
as an upper bound for the number of \mbox{$k$-holes} in $S$.
\end{proof}

\subsection{A lower bound on the minimum number of (general) \mbox{$k$-holes}}
\label{sec:min_lb_gen}
Every set of $k$ points admits at least one polygonization.
Combining this obvious fact with Theorem~6 in~\cite{4HOLES_CGTA}, we obtain the following result.

\begin{theorem}
\label{thm:min}
Let $S$ be a set of $n$ points in the plane in general
position. 
For every $c < 1$ and every $k \leq c \cdot n$, 
$S$ contains $\Omega(n^2)$ \mbox{$k$-holes}.
\end{theorem}

\begin{proof}
 We follow the lines of the proof of Theorem~6 in~\cite{4HOLES_CGTA}.
 Consider the point set $S$ in $x$-sorted order, \mbox{$S=\{p_1, \ldots,
 p_n\}$}, and sets $S_{i,j}=\{p_i, \ldots, p_j\} \subseteq S$.
 The number of sets $S_{i,j}$ of cardinality at least $k$ is
 \[\sum_{i=1}^{n-k+1} \sum_{j=i+k-1}^{n} 1 \ = \ \frac{(n-k+1)(n-k+2)}{2} \ = \ \Theta(n^2).\]
 For each $S_{i,j}$ use the $k-2$ points of $S_{i,j} \backslash \{ p_i, p_j\} $ which are
 closest to the segment $p_ip_j$ to obtain a subset of $k$ points 
 including $p_i$ and $p_j$. 
 Each such set contains at least one \mbox{$k$-hole} which has $p_i$ and $p_j$
 among its vertices.
 Moreover, as $p_i$ and $p_j$ are the left and rightmost points of $S_{i,j}$, they are 
 also the left and rightmost points of %for each of these 
 this \mbox{$k$-hole}. This implies that any \mbox{$k$-hole} of $S$ can count for at most one set $S_{i,j}$, 
 which gives a lower bound of $\Omega(n^2)$ for the number of \mbox{$k$-holes} in $S$.
\end{proof}

A notion similar to $k$-holes is that of islands. %TODO: , first mentioned in~\cite{XXX}
An island $I$ is a subset of $S$, not containing points
of $S\setminus I$ in its convex hull. A $k$-island
is an island of $k$ elements. For example, 
any two points form a $2$-island, and any 3 points spanning an empty triangle 
are a $3$-island of $S$. In particular, convex $k$-holes are also $k$-islands 
while non-convex $k$-holes need not be $k$-islands. In general, any $k$-tuple spans 
at most one $k$-island, while it might span many $k$-holes.
In \cite{RuyClemensIslands12} it was shown that the number of $k$-islands of $S$ is 
always $\Omega(n^2)$ and that the Horton set of $n$ points contains only $O(n^2)$ 
$k$-islands (for any constant $k \geq 4$ and sufficiently large $n$). 

In contrast to this tight quadratic lower bound on the number of islands, 
all (families of) point sets we considered so far contain a super-quadratic 
number of $k$-holes for any constant $k \geq 4$.
For example, the Horton set of $n$ points contains $\Omega(n^3)$ such holes.
Also, at first glance, the result from Theorem~\ref{thm:min} seems easy to improve.
However, note that a general 
super-quadratic lower bound for the number of \mbox{k-holes} for any constant $k$ would
solve a conjecture of B{\'a}r{\'a}ny in the affirmative,
showing that every point set contains a segment that spans a super-constant number of 
\mbox{3-holes}; see e.g.~\cite[Chapter~8.4, Problem~4]{BMP05}. 
This might also be a first step towards proving a quadratic lower bound for the number of convex \mbox{5-holes}.
So far, not even a super-linear bound is known for the latter 
problem~\cite[Chapter~8.4, Problem~5]{BMP05}.  

The results in Section~\ref{sec:min_lb_grid} imply that if there exist point sets %which do not contain a super-quadratic 
with only a quadratic 
number of  \mbox{4-holes}, then they cannot have a grid-like structure. 
In particular, the currently best known configuration for minimizing the numbers of convex $k$-holes, 
the before-mentioned 
squared Horton set, cannot serve as an example.
Due to this and several other observations we state the following conjecture. 

\begin{conj}\label{conj:super}
	For any constant $k\geq 4$, any $n$-point set in general position contains %a superquadratic number
	$\omega(n^2)$ general $k$-holes.
\end{conj}

\subsection{A lower bound on the number of general k-holes in any perturbed grid}
\label{sec:min_lb_grid}
As in Section~\ref{sec:min_ub_horton}, we first deal with the regular integer grid and then derive results for perturbed grids.
We again consider an integer grid~$G$ of size $\sqrt{n} \times \sqrt{n}$. 
As a distance measure we use the $L_{\infty}$-norm, i.e., the distance of two 
points $p,q \in G$ is the maximum of the differences of their $x$- and $y$-coordinates.
The length of a segment $e$ spanned by two points of $G$ is defined by the distance of 
its endpoints.
Similar to prime segments, we denote a $k$-gon in $G$ where all edges are prime segments of $G$ %and bounded solely by prime segments 
as a \emph{prime $k$-gon}. 
Further, let $\square=\partial \CH(G)$ be the square forming the boundary of the convex hull of $G$.
We denote a line $l$ that is spanned by points of $G$ and intersects $\square$ on opposite sides, i.e., $l$ either intersects both vertical segments of $\square$ or $l$ intersects both horizontal segments of $\square$, as \emph{cutting line (of $G$)}.

\begin{obs}
\label{obs:grid_opposite}
For an integer grid $G$ of size $\sqrt{n} \times \sqrt{n}$, consider a prime segment 
$pq$ spanned by two points of $G$ and its supporting line $l$.
If $pq$ has length $d$, then $l$ contains $O(\sqrt{n}/d)$ points of $G$. 
If in addition, $l$ is a cutting line of $G$, then $l$ contains at least $\LF\sqrt{n}/d\RF=\Omega(\sqrt{n}/d)$ points of $G$.
\end{obs}

\begin{theorem}
\label{thm:grid_prime_4holes}
The number of prime 4-holes contained in an integer grid $G$ of size $\sqrt{n} \times \sqrt{n}$ is \linebreak $\Omega(n^2 \log n / \log\log n)$.
\end{theorem}

\begin{proof}
Assume that $\sqrt{n}=3m$ for some integer $m>0$. 
Consider 
contiguous subgrids of $G$ with %equal 
size 
$\sqrt{n}/3 \times \sqrt{n}/3$, and denote by $C$ the central such
subgrid; see Figure~\ref{fig:kh_reg_grid_lines}.
For a fixed point $p\in C$, consider a point $q \in G$ such that $pq$ is a prime segment 
with length $1 \leq d < \sqrt{n}/3$, and its supporting line~$l$.  
Let $l'$ and $l''$ be the two lines parallel to $l$ and spanned by 
points of~$G$ for which no point of~$G$ lies between $l$ and $l'$,
and between $l$ and $l''$, respectively.
Then for each $l^* \in \{l', l''\}$, there exist two disjoint subgrids of $G$ 
(both of size $\sqrt{n}/3 \times \sqrt{n}/3$), for which $l^*$ is a cutting line; see again Figure~\ref{fig:kh_reg_grid_lines} for subgrids where $l'$ is a cutting line.
Thus, due to Observation~\ref{obs:grid_opposite}, 
each of the lines $l'$ and $l''$ contains at least $2\cdot \LF\frac{\sqrt{n}}{3d}\RF = \Omega\left(\frac{\sqrt{n}}{d}\right)$ points of $G$. 
As any point of $G \cap \{l' \cup l''\}$ spans a prime triangle with $pq$ (a triangle where all edges are prime segments), $pq$ is the diagonal of $\Omega\left(\frac{n}{d^2}\right)$ prime $4$-holes.

\begin{figure}[htb]
  \centering
  \includegraphics[page=4, scale=1]{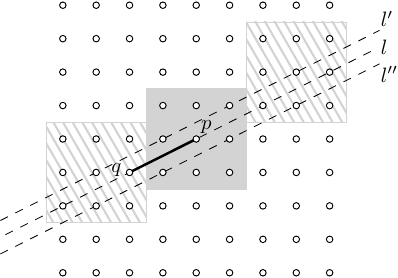} 
  \caption{A prime segment $pq$ with $p\in C$, and the according lines $l$, $l'$, and $l''$. 
	The central subgrid $C$ is drawn in solid gray. Two disjoint subgrids for which $l'$ 
	is a cutting line are drawn with stripes.}
  \label{fig:kh_reg_grid_lines}
\end{figure}

Now consider a fixed point $p\in C$ and a distance $1\leq d < \sqrt{n}/3$.
Then the number of points in $G$ that form prime segments with $p$ can be expressed in terms of Euler's phi-function\footnote{%
Euler's phi-function, also called Euler's totient function, 
$\varphi(d)$ denotes the number of positive integers less than~$d$ that are relatively prime to $d$. 
See for example~\cite{PhiZeta}~(page~52). 
This is the same as the number of segments with endpoints $(0,0)$ and $(a,d)$, $1 \leq a \leq d$, that do not contain any point with integer coordinates in their interior.}.
As $d < \sqrt{n}/3$, this number is exactly $8\cdot\varphi(d)$. 
Thus, summing up over all possible distances, we obtain a lower bound of
\[\sum_{d=1}^{\lfloor\sqrt{n}/3\rfloor-1} 8 \cdot\varphi(d) \cdot \Omega\left(\frac{n}{d^2}\right)
\quad = \quad \Omega\Bigg(n \cdot \sum_{d=1}^{\lfloor\sqrt{n}/3\rfloor-1} \frac{\varphi(d)}{d^2}\Bigg)\] 
for the number of prime 4-holes that have $p$ as a vertex. 
Summing up over all $n/9$ points of $C$, and considering that one 4-hole might 
have been counted four times, namely once for each of its vertices, we obtain 
a lower bound for the total number of prime 4-holes in $G$ of 
\[ \Omega\Bigg( n^2 \cdot \sum_{d=1}^{\lfloor\sqrt{n}/3\rfloor-1} \frac{\varphi(d)}{d^2}\Bigg)\!.\] 

Together with the lower bound of 
$\varphi(d) \geq d / (e^\gamma \log\log d + \frac{3}{\log\log d})$ for $d \geq 3$ where $\gamma$ is Euler's constant ~\cite{bs_ant_1996}~(Theorem 8.8.7)
(and $\varphi(d)=1$ for $d\leq 2$), 
this implies that the number of prime 4-holes in $G$ can be bounded~by
\begin{eqnarray*} 
	\Omega\Bigg( n^2 \cdot \sum_{d=1}^{\lfloor\sqrt{n}/3\rfloor-1} \frac{\varphi(d)}{d^2}\Bigg) 
	& = & \Omega\Bigg(n^2 \cdot \sum_{d=3}^{\lfloor\sqrt{n}/3\rfloor-1} \frac{1}{d\cdot (e^\gamma \log\log d + \frac{3}{\log\log d})}\Bigg)  \\
	& = & \Omega\Bigg(\frac{n^2}{e^\gamma \log\log n + 3} \cdot \sum_{d=3}^{\lfloor\sqrt{n}/3\rfloor-1} \frac{1}{d}\Bigg) \\
	& = & \Omega\Bigg(\frac{n^2 \log n}{\log\log n}\Bigg)\\[-6ex]
\end{eqnarray*}
\end{proof}

For every constant $k\geq 4$, a prime 4-hole can be extended to a prime 
$k$-hole by adding nearby points of the grid. As on the other hand, one prime $k$-hole
contains only a constant number of prime 4-holes, we have the following corollary.

\begin{cor}
\label{cor:grid_prime_kholes}
For any constant $k \geq 4$, the integer grid $G$ of size $\sqrt{n} \times \sqrt{n}$ 
contains $\Omega(n^2 \log n / \log\log n)$ prime $k$-holes.
\end{cor}

In contrast to the bound in Theorem~\ref{thm:ub_min_gen}, the bound in Corollary~\ref{cor:grid_prime_kholes} is independent of $k$. The following alternative bound again depends on $k$.

\begin{theorem}
\label{thm:grid_prime_kholes_alternative}
For every $c < 1$ and every $3\leq k \leq c \cdot 2\sqrt{n}$, the integer grid $G$ of size $\sqrt{n} \times \sqrt{n}$ 
contains %$\Omega(\sqrt{n}^{\LF k/2 \RF+2})=
$\Omega(n^{\LF k/2 \RF/2+1})$ prime $k$-holes.
\end{theorem}

\begin{proof}
	Consider a $k$-hole $H$ for which the following properties hold:
	(i) $H$ is spanned by points in consecutive rows of $G$,
	(ii) the lowest of these rows contains exactly one point of $H$, 
	(iii) the highest of these rows contains one or two points of $H$, and 
	(iv) all rows in-between contain exactly two points of $H$ that are consecutive in the row (but in general not consecutive along the boundary of $H$).
	Figure~\ref{fig:grid_prime_kholes_alternative} shows examples of such $k$-holes for $k \in \{6,7,8\}$. 

\begin{figure}[htb]
  \centering
  	\includegraphics[page=8, scale=1]{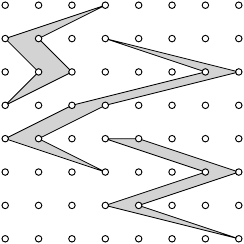} 
  \caption{Examples of prime $k$-holes for the proof of Theorem~\ref{thm:grid_prime_kholes_alternative}. }
  \label{fig:grid_prime_kholes_alternative}
\end{figure}

Clearly, all such $k$-holes are prime. Further, any such $k$-hole contains points from $\LF\frac{k-2}{2}\RF + 2 = \LF\frac{k}{2}\RF + 1$ rows of $G$. 
This gives $\sqrt{n}-\LF\frac{k}{2}\RF$ possible rows for the lowest point of the $k$-hole. 
Further, in every row there are $\sqrt{n}$ or $\sqrt{n}-1$ possibilities to choose the one or two points for the $k$-hole, respectively. This gives a total of at least 
\[ \left(\sqrt{n}-\LF\frac{k}{2}\RF\right)\cdot (\sqrt{n}-1)^{\left(\LF\frac{k}{2}+1\RF\right)} \]
such $k$-holes in $G$. 
For every $c < 1$ and every $3\leq k \leq c \cdot 2\sqrt{n}$, the first factor is $\Omega(\sqrt{n})$.
This implies  a lower bound of $\Omega(\sqrt{n}^{\LF k/2\RF+2})=\Omega(n^{\LF k/2 \RF/2+1})$
on the number of these $k$-holes and thus on the total number of $k$-holes in $G$.
\end{proof}

When comparing the two lower bounds provided in Corollary~\ref{cor:grid_prime_kholes} and Theorem~\ref{thm:grid_prime_kholes_alternative}, we observe that the bound from Theorem~\ref{thm:grid_prime_kholes_alternative} beats the one from Corollary~\ref{cor:grid_prime_kholes} for $k\geq 6$, while Corollary~\ref{cor:grid_prime_kholes} gives a better bound for $k \in \{4,5\}$.

As a last result of this section, we translate these bounds from the integer grid $G$ to grid-like sets of points in general position (for example squared Horton sets). To this end, 
let an $\varepsilon$-perturbation $p_{\varepsilon}(G)$ of~$G$ be a perturbation where every point of~$G$ is replaced by a point at distance at most~$\varepsilon$.
Observe that there exists an $\varepsilon>0$ such that for any $\varepsilon$-perturbation $p_{\varepsilon}(G)$ of~$G$, every prime $k$-hole in~$G$ is also a $k$-hole in~$p_{\varepsilon}(G)$.
This is not true (not even for arbitrarily small $\varepsilon$) 
for a non-prime $k$-hole in~$G$, i.e., a $k$-hole having on its boundary 
points of~$G$ that are not vertices of that $k$-hole.
As in Corollary~\ref{cor:grid_prime_kholes} and Theorem~\ref{thm:grid_prime_kholes_alternative} we count only prime $k$-holes, we obtain the following statement.

\begin{cor}
\label{cor:grid_kholes_lb}
There exists an $\varepsilon>0$ such that any $\varepsilon$-perturbation $p_{\varepsilon}(G)$ of an integer grid $G$ of size $\sqrt{n} \times \sqrt{n}$ contains \\[-3.5ex]
\begin{enumerate} 
	\item[] $\Omega(n^2 \log n / \log\log n)$ 4-holes,\\[-3.5ex]
    \item[] $\Omega(n^2 \log n / \log\log n)$ 5-holes, and\\[-3.5ex]
	\item[] $\Omega(n^{\LF k/2 \RF/2+1})$ $k$-holes for any $6\leq k \leq c \cdot 2\sqrt{n}$ with $c < 1$. 
\end{enumerate} 
\end{cor}

\section{Conclusion}
\label{sec:conclusion}
We have shown various lower and upper bounds on the numbers of 
convex, non-convex, and general \mbox{$k$-holes} and 
\mbox{$k$-gons} in point sets. 
Several questions remain unsettled,
where the maybe most intriguing open question 
is to prove Conjecture~\ref{conj:super}, i.e., to show 
a super-quadratic lower bound for the number of general 
\mbox{$k$-holes} for $k\geq4$.

%\enlargethispage{3ex}
\paragraph{\bf Acknowledgments.} 
Research on this topic was initiated during the \emph{Third Workshop on Discrete Geometry and its Applications} in Morelia (Michoac{\'a}n, Mexico).
We thank Edgar Ch{\'a}vez and Feliu Sagols for helpful discussions.
An extended abstract of part of this article has been presented in~\cite{KHOLES_CCCG}.
%Research of Oswin Aichholzer and \mbox{Birgit} Vogtenhuber supported by the ESF EUROCORES programme EuroGIGA -- CRP `ComPoSe', Austrian Science Fund (FWF): I648-N18.
%Research of Ruy Fabila-Monroy partially supported by CONACyT (Mexico), grant 153984.
%Research of Thomas Hackl supported by the Austrian Science Fund (FWF): P23629-N18 `Combinatorial Problems on Geometric Graphs'.
%Research of Clemens Huemer partially supported by projects MTM2012-30951 and Gen. Cat. DGR 2009SGR1040. 
%Research of Marco Antonio Heredia, Hern{\'a}n Gonz{\'a}lez-Aguilar, and Jorge Urrutia partially supported by CONACyT (Mexico), grant CB-2007/80268.
%Work by Pavel Valtr supported by project 1M0545 of the Ministry of Education of the Czech Republic. 

%---------------------------- Bibliography -------------------------------
{\small
\bibliographystyle{abbrv}
\bibliography{holesbib}
}
\end{document}